\documentclass[11pt]{article}
\usepackage[T1]{fontenc}
\usepackage[utf8]{inputenc}
\usepackage{amsmath,amssymb,amsthm}
\usepackage{booktabs,tabularx,longtable,array}
\usepackage[margin=1in]{geometry}
\usepackage{tikz}
\usetikzlibrary{arrows.meta,positioning}
\usepackage{xurl}
\usepackage[hidelinks]{hyperref}
\newtheorem{proposition}{Proposition}

\newcommand{\auth}{\mathsf{authorize}}
\newcommand{\sha}[1]{\path{#1}}

\hypersetup{pdftitle={Omission-Safe Conditional Contracts for Evidence-Closed CI Qualification},pdfauthor={Zhiheng Zhang}}
\title{Omission-Safe Conditional Contracts for Evidence-Closed CI Qualification}
\author{Zhiheng Zhang\\School of Computing and Information Systems\\The University of Melbourne\\Melbourne, Australia\\\href{mailto:zhiheng0mera@gmail.com}{zhiheng0mera@gmail.com}\\\href{https://orcid.org/0009-0005-9005-9931}{ORCID: 0009-0005-9005-9931}}
\date{}
\begin{document}
\maketitle
\begin{abstract}
Continuous-integration evidence is conditional: a workflow determines which jobs should run, while an archive records which jobs remain observable. Reconstructing requirements from those observations lets an omission erase its own obligation. We present a qualification method that compiles a pinned workflow into conditional job requirements independently of target presence, retains a fixed capture inventory, and checks source-complete derivation certificates with a separate implementation. For a fixed source, context, scope and binding universe, with sound obligations and applicability judgments that preserve established truth under information extension, withdrawing observation facts cannot create authorization. The implementation handles static matrices, dependency guards and bounded pinned reusable workflows, preserving unknown applicability and separating tolerated outcomes from evidence coverage. Both implementations match the 45 hand-specified authorization controls, reject 180 source-certificate attacks, and exhibit no authorization gain over 704 distinct proper development erasures. An explicitly witnessed 12-job platform run qualifies in both implementations; all 4,095 proper payload-subset erasures block authorization. A frozen collection spans 240 runs in eight workflow families across six projects. All 160 evaluation source certificates independently validate, but unavailable historical execution-source identities leave actual qualification open. Initial gate diagnoses agree on 155/160 runs; a preserved unknown-propagation repair yields 160/160 agreement in a retrospective rerun. The study establishes bounded source extraction and omission resistance, rather than natural-fault prevalence or producer authenticity.
\end{abstract}

\section{Introduction}
A successful CI record describes an observed outcome. Using it to justify a claim also requires deciding what execution the claim demanded. In a conditional workflow, that requirement depends on source-level guards, matrix expansion, reusable definitions and predecessor results. These dependencies are lost when an evidence consumer treats the jobs present in an archive as the complete set of jobs that mattered.

Consider a workflow that builds two configurations and then checks their results. If one configuration is absent, a consumer that rebuilds its requirements from surviving records can accept the remaining success. A different omission can remove one of two identically named records and make an ambiguous binding appear unique. Neither error is repaired by validating the internal consistency of the remaining records. The requirement and the binding alternatives must survive independently of their payloads.

Evidence-Closed Qualification separates evidence closure, gate satisfaction and claim authorization. This study develops its conditional CI front end: requirements are compiled from a workflow rather than supplied only as a static gate list. Our method separates three objects: the obligation universe compiled from source, the inventory fixed at capture, and the facts currently available about each captured record. Conditional applicability is evaluated with explicit unknowns. False applicability discharges execution coverage without claiming that a test ran; unavailable applicability blocks authorization. A tolerated failure still requires an observed job and contributes the platform's accepted result to downstream scheduling.

The contributions are a source-derived conditional qualification contract with a fixed-inventory erasure semantics; a bounded compiler and separately authored source-completeness checker; and an evaluation combining hand-specified controls, source-certificate attacks, exhaustive bounded erasures, platform conformance observations and a prospectively selected multi-project archive. The archive study distinguishes extraction from execution-source identification: a complete collection of job outcomes does not itself establish which source and event context governed their execution.

\section{Relation to Existing Methods}
Continuous compliance and assurance tooling connect changing development observations to persistent requirements \cite{kellogg2020continuouscompliance,sorokin2024continuous}. The Evidential Tool Bus already supports Datalog-based assurance workflows, evidence hashing and dependency-driven maintenance. Goal-directed ASP analyzes semantic consistency, adequacy and defeaters in assurance cases \cite{murugesan2024semantic}. These are substantial precedents for executable assurance reasoning. Our narrower decision concerns whether withdrawal of CI observations can change the obligations against which those same observations are judged.

in-toto fixes an independently declared supply-chain layout and verifies authorized functionaries and artifact rules against link metadata \cite{torresarias2019intoto}. That separation motivates keeping requirements outside observed payloads. PROV constraints and SHACL express consistency and validation rules over structured facts \cite{w3c2013provconstraints,w3c2017shacl}. The contribution here is a conditional CI front end and its capture boundary, rather than a new general constraint language or cryptographic authentication mechanism.

Static workflow validation addresses a related but different input. Horber and Sotiriadis combine syntax, schema, registry and execution-graph checks, evaluating 843 workflow configurations and a labeled comparison with actionlint \cite{horber2026validation}. Their DAG and dependency checks precede execution; their reported scope limits runtime-dependent expressions and cross-workflow resolution. Our checker reads a supplied source closure together with archived observations, retains unresolved historical context, and checks that the observation set has not silently determined the requirement set. Syntax-valid workflows can therefore yield open qualification, and well-formed records of failure can yield closed evidence with an unsatisfied claim.

\section{Conditional Qualification Contracts}
\subsection{Inputs and trust boundary}
Write the fixed input boundary as
\[
 B=(W,K,S,J,I,A),
\]
where $W$ is the supplied source closure, $K$ the immutable context roots, $S$ the declared claim scope, $J$ the source-derived job universe, $I$ the capture inventory and $A$ the capture and source-identity assertions. The source closure includes the root workflow and supplied reusable definitions at immutable revisions. An inventory entry records an observed record identity, display name and run binding, including entries that are unmatched or ambiguous. The available payload facts form a well-formed partial map $x$ over identities fixed in $I$; malformed boundary metadata is rejected rather than interpreted as an observation erasure.

The boundary is relative to the collector. An identity assertion or a digest of an API response does not authenticate the producer, prove the collector saw every event, or establish the historical root workflow revision. In particular, a PR head commit can differ from the merge revision used for a local reusable workflow. Capture completeness, event-context completeness and claim scope are explicit inputs fixed before erasure; changing one is a boundary amendment. A cancellation assertion must describe the context at guard evaluation; an arbitrary final run status does not establish that historical context.

\begin{figure}[t]
\centering
\begin{tikzpicture}[node distance=8mm and 9mm,box/.style={draw,rounded corners,align=center,text width=35mm,minimum height=11mm},>=Stealth,font=\small]
\node[box] (source) {Pinned source\\and context roots};
\node[box,right=of source] (slots) {Conditional slots\\and source witnesses};
\node[box,below=of source] (capture) {Fixed inventory\\and binding alternatives};
\node[box,right=of capture] (facts) {Partial payload facts\\including unknowns};
\node[box,right=of slots,yshift=-10mm] (decision) {Independent checks\\closure / gate / authority};
\draw[->] (source)--(slots);
\draw[->] (capture)--(facts);
\draw[->] (slots)--(decision);
\draw[->] (capture.south) -- ++(0,-6mm) -| ([xshift=4mm]decision.south);
\draw[->] (facts)--(decision);
\end{tikzpicture}
\caption{Fact withdrawal changes the lower observation path, not the source requirement universe or capture binding alternatives.}
\label{fig:boundary}
\end{figure}
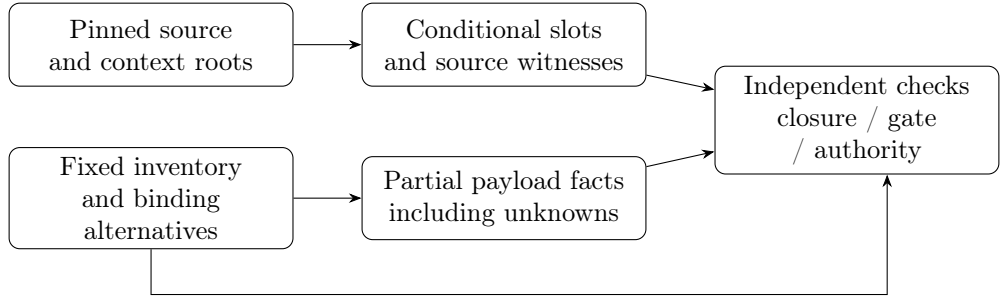

\subsection{Source expansion and applicability}
The compiler parses YAML 1.2, retains source order, checks the needs DAG and expands bounded static matrices. Ordered include augmentation is applied against original axis assignments; exclude entries remove matching assignments. Names use original axis order and evaluate caller and callee display expressions separately. Include-only matrices and metadata augmentation are distinct: extra metadata does not automatically become an implicit name suffix.

Pinned reuse is expanded hierarchically. Wrapper activation precedes child activation, child dependency names are namespaced, and aggregated child results supply downstream needs. Input defaults and declared types are checked. Known callee source is retained even when an unavailable or ill-typed input leaves its names or activation unresolved. Unavailable source, unsupported dynamic expansion and depth overflow remain blocking unknowns; they do not justify dropping a callee from the contract. Recursion is bounded to three reusable levels and each matrix to 256 cells.

The expression subset includes typed scalar literals, dotted or bracketed lookup, equality, negation, operand-returning conjunction and disjunction, and status functions. The concrete rules follow GitHub's coercion and status behavior \cite{github2026expressions,github2026syntax}. Partial evaluation distinguishes a missing property of a known complete object from an unavailable historical fact. Our live controls observe null inside expressions and empty text after rendering: the missing scalar compares equal to the string \texttt{'0'}, while a literal empty string does not. The implementation retains that typed distinction; unavailable historical facts remain unknown. Generated needs outputs remain partial even when an external event object was completely captured.

A partial scalar is reported concrete only when every compatible completion has that value. For example, an unknown left operand conjoined with false has false truthiness, but can return an erased empty string rather than the scalar false. The implementation conservatively retains an unknown scalar in such cases. This loses precision, while preserving the information-order condition used below.

Job guards precede matrix expansion in scheduling semantics. A false matrix group can have one unexpanded skipped placeholder rather than one record per potential cell. The observer permits absence or a uniquely attributable skipped group placeholder. It never expands that placeholder into successful executions. Unknown guards cannot be replaced by false. An owned platform probe establishes the bounded reusable reduction used here: a nonempty all-skipped child set yields a skipped caller result, while successful and skipped children yield success. An unresolved child remains unknown, so a downstream default-success guard cannot waive coverage by guessing that reduction. An implicit success condition applies when no actual status-function call overrides it; quoted text resembling a function call does not override it.

Own-property lookup follows ASCII case-insensitive context keys, retaining unknown for ambiguous captures or unavailable roots. The bounded profile uses finite scalars with magnitude at most $2^{53}-1$, shared decimal conversions and unsigned lowercase hexadecimal/octal string conversions up to $2^{31}-1$, rather than assuming that JavaScript and runner conversions coincide \cite{github2026implementations}. Unsupported numeric strings, nonidentical non-ASCII case comparisons and unresolved structured or nonintegral display names remain unknown. Failure tolerance requires an evaluated Boolean policy. Typed inventory identities, claim scopes and completeness assertions are checked before authorization. These restrictions preserve partial-information judgments without claiming full GitHub expression-language equivalence.

\subsection{Capture binding and decision axes}
Expected names are derived from source; observed names are retained in $I$. Binding requires a unique source slot and a unique capture candidate. Payload withdrawal cannot delete an entry from $I$ or resolve its ambiguity. Every captured record retains coverage, identity, state and source-attribution obligations, including records of inactive groups. A success attached to an inactive group is a contradiction; erasing its payload leaves an unresolved captured obligation.

Evidence closure is \textsc{Closed} when structural obligations are established and \textsc{Open} otherwise. Gate satisfaction is a three-valued conjunction of accepted local outcomes and the scope's non-vacuity predicate. An established rejecting outcome yields \textsc{Unsat}; an unresolved outcome or applicability yields unknown unless another established rejection already decides the conjunction. Source-identity uncertainty can leave evidence open even when all supplied local outcomes satisfy the candidate contract. Authorization requires both closed evidence and a satisfied gate.

Two scopes are explicit. Conformity to the job contract can be satisfied when all groups are inapplicable. A non-vacuous execution claim additionally requires at least one established active leaf in its declared job scope. An accurate all-skipped archive can therefore be closed while failing that execution claim. The interface names this mode \texttt{test-execution}; it denotes job-level non-vacuity, not identification or execution of individual test steps. A job-level claim establishes no fact about arbitrary shell commands or whether each test step inside a successful job actually ran.

\subsection{Source-complete certificates}
Each compiled leaf carries the hash of its source bytes, a path to the parsed job node, its matrix index and the caller chain for reusable expansion. The certificate also includes source identities, guards, dependency aliases, input bindings and complete group/instance lists. These are dereferenceable derivation inputs rather than identifiers asserted to be JSON pointers. Byte offsets are diagnostic metadata.

The independent checker reparses the supplied YAML closure and reconstructs the full universe and relevant policies. It checks the semantic projection of the certificate, including duplicate source-witness lists and caller provenance. Comparing only the submitted witness count would permit a producer to omit the requirement and its witness together. Recompilation prevents a forged certificate from narrowing the source obligation universe. Claim scope remains a separately fixed boundary input; the source certificate does not authenticate it. These checks establish consistency with the supplied source boundary, not authenticity of that boundary.

\section{Omission Safety}
For fixed $B$, define $x\preceq y$ when every concrete payload coordinate of $x$ occurs identically in $y$. Unknown is the least informative coordinate. An erasure $e$ satisfies $e(x)\preceq x$; it can withdraw scalar facts or entire payloads while retaining inventory tombstones. Editing a value, deleting a capture candidate, changing completeness assertions or narrowing claim scope is outside this order.

\begin{proposition}[Fixed-boundary omission safety]
Suppose source enumeration and binding completeness are independently checked; positive obligation judgments are sound and preserve established truth as information increases; applicability guards are sound and preserve either established truth value; unknown obligations cannot authorize; and qualification does not narrow scope from surviving evidence. Then
\[
 e(x)\preceq x\ \wedge\ \auth_B(e(x))\quad\Longrightarrow\quad\auth_B(x).
\]
\end{proposition}
\begin{proof}
Every completion compatible with $x$ is also compatible with $e(x)$. An obligation established under $e(x)$ remains semantically valid under $x$, and the persistence assumption ensures that the evaluator retains its established status. A guard established false under the erased observations is false in every such completion; a guard established true remains true. Each active slot's established binding, coverage and accepted outcome persists, as do the independent source and capture obligations. Any established active-leaf witness for non-vacuity persists. Thus every conjunct required for authorization under the erased observations also holds under the more informative observations.
\end{proof}

Soundness alone is insufficient for the algorithmic statement: a sound evaluator could report true on a less informative input and unknown on a more informative one. The persistence assumption excludes that loss of an already established judgment. The fixed boundary is also necessary. If requirements are rebuilt from survivors, deleting an absent-required-job indication can delete that job's requirement. If binding candidates are rebuilt from survivors, deleting one of two same-name candidates can turn ambiguity into uniqueness. If unknown applicability is treated as false, removing the deciding predecessor fact can waive an execution obligation. Persistent inventories address the second error; source-derived slots and unknown blocking address the first and third. The proposition supplies a method invariant, while finite tests check whether the implementation respects it on particular inputs.

\section{Evaluation Design}
\subsection{Implementations and controls}
The native compiler/observer uses Node.js and the independent reference/checker uses Python with a separate YAML parser. Checking predicates are not imported between implementations. Shared schema and serialization conventions permit comparison of source-derived contracts. Runtime versions, exact dependency locks, source hashes and raw outputs accompany the artifact.

The hand-specified development controls use job-level YAML fragments to cover active success and failure, cancellation, absent jobs, false and unknown guards, tolerated failure, needs propagation, matrix expansion and naming, duplicate and extra records, source/run mismatches, generated context completeness and reusable expansion. They are abstract contract controls, not separately runnable CI deployments. Labels state authorization independently of either implementation. Regressions discovered by adversarial inspection remain identified as development cases; live workflows supply the platform conformance evidence. The final control set has 45 cases (14 authorizable and 31 rejecting), with 11 independently specified closure/gate pairs. The erasure procedure executes 794 checks over 749 distinct base/input pairs: 45 original inputs and 704 proper erasures. The repeated identity transformations are not counted as additional cases. All generator amendments precede the first evaluation-run analysis.

We compare three authored front-end ablations using the same native qualification backend: treating every job as unconditionally required; rebuilding source groups and capture inventory from surviving jobs; and treating unknown guards as inactive. These deliberately isolate information discarded by the front end. They are not estimates of production-tool prevalence or comparisons with tools that implement the full declared contract. The separately authored full reference is the principal semantic comparison.

\subsection{Prospective archive and platform observations}
The collection protocol fixes eight build/test workflow families across six projects, a 2026-10-27 UTC cutoff and within-family success/failure strata. It selects five successful and five failed development runs, and ten successful and ten failed evaluation runs per family. All preflight run IDs are excluded. Each exact run attempt is collected through paginated jobs endpoints, with immutable source snapshots, referenced reusable definitions and raw response hashes. Empty archives and API errors remain selected.

The 240 source runs comprise 80 development and 160 evaluation runs. All 34 complete zero-job archives are failures; their cause is not established by the archive. One selected success has an incomplete jobs API response and remains in the analysis. Six repository/head revisions occur in both splits, so the evaluation is disjoint by run, not fully held out by commit. Related workflow families are grouped under their six projects.

Root execution revisions and full historical event payloads remain unavailable for all 240 runs under the declared source-identity criterion. Consequently the study reports candidate-source enumeration, mapping and independent checking separately from actual qualification. It does not label a workflow success as gold authorization or fill missing context from the jobs that happened to survive. Owned platform probes directly exercise concrete expression coercions, skipped matrices, tolerated failure and reusable input behavior without running collected third-party commands.

\subsection{Measurements}
Controls measure authorization errors, bounded erasure gains, certificate attack rejection and agreement between implementations. Archive measurements report source instances, unique leaf-name matches, unsupported groups, certificate disagreements and actual authorization while preserving unavailable identity facts. Record and transformation counts are distinct from project counts.

Resource measurements separate source compilation, qualification including defensive source reconstruction, and independent certificate checking. Static matrices contain 1, 10, 100 or 256 cells, with five warmups and thirty measured rounds per runtime/configuration under a 180-second configuration budget. Completed cells supply medians; incomplete cells would be censored. Repeated rounds quantify this host's execution cost, not a population speed advantage or an algorithmic comparison between programming languages.

\section{Results and Mechanisms}
\subsection{Controlled qualification and certificate integrity}

Both full implementations match all 45 specified authorizations and all 11 specified closure/gate pairs. The independent checker accepts all genuine control certificates and rejects all 180 attacked certificates. No proper erasure in the 704-input development domain creates authorization. These are finite development checks of the declared invariant.

\begin{table}[t]

\centering

\begin{tabular}{lrr}

\toprule

Method & False authorization / 31 & False blocking / 14 \\

\midrule

Native full profile & 0 & 0 \\

Independent full reference & 0 & 0 \\

Unconditional requirements & 1 & 5 \\

Survivor-derived universe & 12 & 2 \\

Unknown treated inactive & 3 & 0 \\

\bottomrule

\end{tabular}

\caption{Authorization disagreements on hand-specified development controls. The three lower rows are front-end ablations with the same native backend.}

\label{tab:controls}

\end{table}

The ablations distinguish different mechanisms. Unconditional requirements reject legitimate inactive groups and can accept an execution that contradicts its actual false guard. Survivor-derived requirements can remove both a missing slot and its obligation, while rebuilding inventory can remove a conflicting candidate. Unknown-as-inactive removes an unresolved requirement. The complete profiles retain those source and capture obligations.

\subsection{Candidate-source extraction and unavailable execution identity}

The 160 evaluation runs contain 2,453 captured job records. Both compilers produce source-check-compatible certificates for all 160 candidate snapshots, with 2,859 potential leaf instances. There are 2,201 uniquely matched leaf-name records (89.7\%); this metric excludes group placeholders and is not a complete qualification rate.

\begin{table}[t]

\centering\small

\begin{tabular}{lrrrr}

\toprule

Workflow family & Runs & Records & Unique leaf matches & Unsupported runs \\

\midrule

HTTPX test-suite & 20 & 100 & 100 & 0 \\

Fastify CI & 20 & 260 & 248 & 0 \\

Fastify integration & 20 & 24 & 24 & 0 \\

Undici CI & 20 & 522 & 522 & 0 \\

Undici historical nodejs & 20 & 400 & 160 & 20 \\

Flask tests & 20 & 261 & 261 & 0 \\

Requests tests & 20 & 470 & 470 & 0 \\

pytest tests & 20 & 416 & 416 & 0 \\

\bottomrule

\end{tabular}

\caption{Evaluation split; candidate-source mapping rather than execution-source authentication. Related families share a project.}

\label{tab:archives}

\end{table}

The 252 records without a unique leaf-name match have two observed causes. Twelve are the unexpanded Fastify matrix placeholders, six each for unit and compatibility groups. The other 240 are historical Undici test leaves whose numeric matrix values are supplied to string-typed reusable inputs in the candidate snapshot. The compiler retains the complete callee universe with unknown inputs, rather than deriving names from the observed suffixes. All 20 evaluation runs in that family retain an unsupported-input obligation.

All 240 collected runs lack an independently witnessed historical root execution revision under the declared criterion; the 160 evaluation runs therefore remain open and unauthorized in both implementations. The successful source-certificate checks validate extraction from supplied candidate snapshots. They do not upgrade those snapshots into established execution configurations.

The first frozen evaluation gives agreement on closure and authorization in all 160 runs, and on the gate in 155/160. The five discrepancies occur in Fastify CI: unavailable cancellation context makes upstream reusable guards unknown, but the native observer propagated the captured successes as known needs results. That incorrectly activated downstream coverage requirements and produced an unsatisfied gate where the reference retained unknown applicability. The amendment preserves an unknown generated result whenever applicability is unknown. A retrospective rerun agrees on all 160 gates without changing source selection, reference semantics or authorization. The original results, method snapshot and five counterexamples remain archived; the repaired rerun is not a second blinded evaluation.

Subsequent conformance hardening checks ASCII-insensitive property names, bounded numeric conversions, typed capture identities and frame assertions, Boolean failure policies, and reusable caller aggregation. These are separately preserved retrospective amendments. Their final replay retains all evaluation source certificates and decisions; they do not replace the original frozen result with a new blinded claim.

\subsection{Witnessed platform behavior and whole-record erasure}

The owned conformance workflow supplies an explicit workflow revision witness in addition to API job metadata. Its original capture has 12 records, eight active leaves and four inactive placeholders. Both implementations classify it as closed, satisfied and authorized, and independently validate its source certificate. The raw tolerated job remains a failure, while its dependent success condition runs and its failure condition is skipped. The matrix guard produces one skipped placeholder before expansion. Missing-property probes observe null within expressions and empty text after rendering, preventing a text-level missing-value convention from being used as a scalar type rule.

\begin{table}[t]

\centering

\begin{tabular}{lrrl}

\toprule

Payload subset & Count & Authorized & Closure / gate \\

\midrule

Original capture & 1 & 1 & Closed / Sat \\

Only inactive payloads withdrawn & 15 & 0 & Open / Sat \\

Any active payload withdrawn & 4080 & 0 & Open / Unknown \\

\bottomrule

\end{tabular}

\caption{All $2^{12}$ payload subsets retain the original capture inventory. Both implementations give the same axes.}

\label{tab:platform}

\end{table}

The 15 inactive-only erasures preserve the execution gate but leave captured consistency obligations unresolved. The 4,080 subsets withdrawing an active payload also lose an execution fact. Thus withdrawing an inactive record and declaring that no such record was ever captured are observably different inputs. The latter changes the trust boundary.

Windows and Linux replay at the frozen method revision reproduce the same 45 control decisions, 180 attack rejections and 794 enumerated erasure checks. These are repeated executions of the same cases, not additional independent samples. The exhaustive owned-run study is a controlled platform case, distinct from the six-project archive.

Platform observations also distinguish an implementation assumption from the actual runner: binary numeric strings and uppercase hexadecimal prefixes do not behave like JavaScript Number conversion, while large hexadecimal conversion follows the runner integer domain. Mixed-case context properties do resolve. The final profile retains unknown outside its verified conversion domain instead of allowing an unsupported comparison to manufacture an inactive guard. The subset replay reuses only pure compilation of the identical fixed frame; checking predicates are unchanged, cached results are checked against uncached samples, and resource timings use the uncached drivers.

\subsection{Bounded execution cost}

\begin{table}[t]

\centering\small

\begin{tabular}{llrrrr}

\toprule

Runtime & Cells & Compile & Qualify & Certificate check & Measured rounds \\

\midrule

Node native & 1 & 0.398 & 0.437 & -- & 30 \\

Node native & 10 & 0.461 & 0.614 & -- & 30 \\

Node native & 100 & 1.705 & 3.510 & -- & 30 \\

Node native & 256 & 3.031 & 9.547 & -- & 30 \\

Python reference & 1 & 2.903 & 2.966 & 3.070 & 30 \\

Python reference & 10 & 4.410 & 4.622 & 4.874 & 30 \\

Python reference & 100 & 20.021 & 25.089 & 23.773 & 30 \\

Python reference & 256 & 45.835 & 72.029 & 54.513 & 30 \\

\bottomrule

\end{tabular}

\caption{Median milliseconds on one Windows 11 / Intel i7-14650HX host. Qualification includes defensive source recompilation; stage medians are not additive.}

\label{tab:cost}

\end{table}

All eight cells complete their 30 measured rounds after five warmups, yielding 240 measured runtime/configuration rounds without censoring. The runtimes are Node.js 24.14.0 with yaml 2.8.3 and Python 3.14.3 with ruamel.yaml 0.18.16. Independent certificate reconstruction adds a separately measured parsing/checking stage. Differences include parser and implementation choices; they do not establish an intrinsic algorithmic advantage or cross-machine speedup.

\section{Scope and Limitations}
The result is relative to fixed source, context and capture assertions. A malicious collector can omit a record before inventory capture; hashes and the omission proposition do not detect that omission. A different authenticated capture mechanism would establish a different boundary. The public archive's unavailable execution revisions illustrate why internally consistent outcomes alone cannot establish source closure.

The implemented language is a bounded subset of GitHub Actions. Dynamic matrices, unsupported functions and unavailable inputs retain unknown obligations. Conservative scalar abstraction can reject archives whose guard truthiness could be decided by a richer domain. Job-name binding can remain ambiguous after source expansion, especially when distinct jobs share display names. These cases limit usable qualification without turning uncertainty into authorization.

The hand controls and adversarial regressions establish particular behaviors, not a natural-fault detection rate. The multi-project collection broadens source and job structure beyond single-job documentation workflows, but its prospective stratification is not a random sample of CI systems. Its temporal drift, related workflow migrations and revision overlap limit independent transfer claims. Two-OS replay checks runtime portability of the supplied cases; it does not add independent scientific samples.

\section{Conclusion}
Conditional CI evidence requires a requirement set that survives independently of its observations. Source-derived slots, persistent capture alternatives and sound unknown handling make that separation executable and yield a fixed-boundary omission-safety property. Independent source reconstruction checks that a certificate has not removed the obligation it was meant to justify. The archive findings distinguish successful extraction from established execution identity, preserving the information needed to decide whether qualification is actually closed.

\bibliographystyle{plainurl}
\bibliography{references,additional}
\end{document}